\documentclass[11pt]{article}
\usepackage{fullpage}
\usepackage[utf8]{inputenc}
\usepackage[T1]{fontenc}
\usepackage[american]{babel}
\usepackage{braket}
\usepackage{amsthm}
\usepackage{hyperref}
\usepackage{amsmath, amssymb}
\usepackage[pdftex]{graphicx}
\usepackage{authblk}
\usepackage[dvipsnames]{xcolor}
\usepackage{algorithm}
\usepackage{mathtools}
\usepackage{algpseudocode}
\usepackage{latexsym}
\usepackage{microtype}
\usepackage{tikz}
\usetikzlibrary{quantikz2}

\hypersetup{
    linktocpage,
    colorlinks=true,
    citecolor=blue,
    linkcolor=Magenta,
    urlcolor=orange,
}

\usepackage{complexity}
\let\H\relax

\usepackage{ammacros}
\newcommand{\Q}{\mathsf{Q}}

\usepackage[capitalise,nameinlink]{cleveref}
\Crefname{lemma}{Lemma}{Lemmas}
\Crefname{fact}{Fact}{Facts}
\Crefname{theorem}{Theorem}{Theorems}
\Crefname{corollary}{Corollary}{Corollaries}
\Crefname{claim}{Claim}{Claims}
\Crefname{example}{Example}{Examples}
\Crefname{problem}{Problem}{Problems}
\Crefname{definition}{Definition}{Definitions}
\Crefname{notation}{Notation}{Notations}
\Crefname{assumption}{Assumption}{Assumptions}
\Crefname{subsection}{Subsection}{Subsections}
\Crefname{section}{Section}{Sections}
\Crefformat{equation}{(#2#1#3)}
\Crefname{figure}{Figure}{Figures}

\newtheorem{theorem}{Theorem}
\newtheorem{definition}{Definition}
\newtheorem{lemma}[theorem]{Lemma}

\newtheorem{fact}{Fact}

\newtheorem{proposition}[theorem]{Proposition}

\usepackage{thmtools}
\usepackage{thm-restate}

\makeatletter
\def\E{\@ifnextchar[{\@withb}{\@withoutb}}
\def\@withb[#1]{\mathop{\mathbb{E}}_{#1}\brk}
\def\@withoutb{\mathop{\mathbb{E}}\brk}
\makeatother

\makeatletter
\def\Pr{\@ifnextchar[{\@witha}{\@withouta}}
\def\@witha[#1]{\mathop{\operator@font Pr}_{#1}\brk}
\def\@withouta{\mathop{\operator@font Pr}\brk}
\makeatother

\title{Maximum Separation of Quantum Communication Complexity \\ With and Without Shared Entanglement}

\author[1]{Atsuya Hasegawa\thanks{atsuya.hasegawa@math.nagoya-u.ac.jp}}
\author[1]{Fran{\c{c}}ois Le Gall\thanks{legall@math.nagoya-u.ac.jp}}
\author[2]{Augusto Modanese\thanks{augusto.modanese@cispa.de}}

\affil[1]{\textit{Graduate School of Mathematics, Nagoya University, Japan}}
\affil[2]{\textit{CISPA Helmholtz Center for Information Security, Germany}}

\date{}

\begin{document}

\maketitle

\begin{abstract}
    We present relation problems whose input size is $n$ such that they can be solved with no communication for entanglement-assisted quantum communication models, but require $\Omega(n)$ qubit communication for $2$-way quantum communication models without prior shared entanglement. This is the maximum separation of quantum communication complexity with and without shared entanglement. To our knowledge, our result even shows the first lower bound on quantum communication complexity without shared entanglement when the upper bound of entanglement-assisted quantum communication models is zero. Our result refutes a quantum analog of Newman’s theorem.
    
    The problem we consider is parallel repetition of any non-local game which has a perfect quantum strategy and no perfect classical strategy, and for which a parallel repetition theorem holds with exponential decay.
\end{abstract}

\section{Introduction}

\subsection*{Background}

The power of shared entanglement has been a central attention in the field of quantum information processing. Especially it has been investigated from the perspectives of quantum non-locality and non-local games \cite{bell1964einstein,clauser1969proposed}, and communication complexity \cite{yao1979some,yao1993quantum}. A non-local game is a game in which two cooperating players, Alice and Bob, each receives a question from a referee, and then respond with an answer. The referee randomly selects the questions according to a known distribution, and, upon receiving answers from Alice and Bob, decides whether they win or lose. In a quantum setting, Alice and Bob are allowed to share entanglement before receiving questions. In communication complexity, we consider the setting where inputs are distributed to Alice and Bob, and they coordinate to solve some problems depending on distributed inputs. The goal is to identify the minimum amount of communication between them. In quantum communication complexity, we also consider shared entanglement to reduce communication as an analog of shared randomness in classical communication complexity \cite{cleve1997substituting}. We refer to \cite{BBT05,BCMdW10} for surveys of these concepts and their relations.

\subsection*{Our results}

In this paper, we show there exist relation problems which can be solved with no communication for entanglement-assisted quantum communication models and hence non-signaling communication models, and require $\Omega(n)$ qubits communication for $2$-way quantum communication models without prior entanglement. The communication task we consider is parallel repetition of a non-local game which has a perfect quantum strategy and no perfect classical strategy, and for which a parallel repetition theorem \cite{Raz98,Hol09} holds with exponential decay. Here, a perfect quantum (resp. classical) strategy means a quantum (resp. classical) strategy for non-local games which wins with probability 1, and the parallel repetition theorem with exponential decay states that for a two-prover game with value smaller than 1, parallel repetition reduces the value of the game in an exponential rate. The parallel repetition of the magic square games \cite{Mer90,Mer93,Ara02} is the canonical example of our communication tasks. Let us denote by $\Q_2^{pub}(R)$ the quantum 2-way communication complexity without shared entanglement and with shared randomness (or public coin) for a relation $R$, and by $\Q^{*}(R)$ the quantum communication complexity with shared entanglement for $R$. Here is our result.

\begin{restatable}{theorem}{restateMainThm}\label{thm:main}
    There exists a relation $R$ such that $\Q^{*}(R)=0$ and $\Q_2^{pub}(R)=\Theta(n)$.
\end{restatable}

Moreover, in the theorem above, the upper bound holds with zero error and the lower bound holds with $2^{-o(n)}$ error as well as the bounded error. Since there exists a trivial upper bound of $O(n)$ in communication complexity, this is the maximum separation of quantum communication complexity with and without shared entanglement asymptotically.

It is known that similar relations can be constructed from a communication problem which is easy ($O(\poly (\log n))$ communication) for quantum one-way (and thus SMP) communication complexity and hard ($\Omega(\poly (n))$ communication) for classical communication complexity. It can be shown that such relations are solved from shared entanglement of $O(\poly (\log n))$ qubits and no additional communication, and they are also hard for classical communication protocols by a reduction. These constructions from communication complexity to non-local relations are explicitly written in Section 4 of \cite{BCMdW10}. 
More recently, Jain and Kundu \cite{jain2021direct} showed the classical communication complexity of the parallel repetition of magic square games is $\Omega(n)$.
To our knowledge, our result is the first lower bound on quantum communication complexity without prior entanglement when we keep the upper bound of entanglement-assisted quantum communication models zero. See \cref{tab:result} for a comparison between our result and previous results.

\begin{table}[htbp]
    \centering
    \begin{tabular}{cc|c}
         & Relation problem & Lower bound \\
         \hline \hline
         & Non-local DJ \cite{BCT99} & $\Omega(n)$ in $1$-way deterministic model \\ \hline
         & Non-local HM \cite{BBT05,gavinsky2006bounded} & $\Omega(\sqrt{n})$ in $1$-way randomized model \\ \hline
         & \begin{tabular}{c} Non-local relations generated \\ from \cite{Gav21,RGT22,GGJL24} \end{tabular} & $\Omega(\poly (n))$ in $2$-way randomized model \\ \hline
         & \begin{tabular}{c} Parallel Magic Square Games \cite{jain2021direct} \end{tabular} & $\Omega(n)$ in $2$-way randomized model \\ \hline
         & \begin{tabular}{c} Parallel Magic Square Games \\(\cref{thm:main}) \end{tabular} & $\Omega(n)$ in $2$-way \emph{quantum} model \\
    \end{tabular}
    \caption{Relation problems solved by the entanglement-assisted communication model with zero communication. We abbreviate Deutsch-Jozsa as DJ and Hidden Matching as HM.}
    \label{tab:result}
\end{table}

We complement the main result by showing that the largest separation does not hold for a function rather than a relation.

\begin{restatable}{theorem}{restateCompThm}\label{thm:function}
    Let $F: X \times Y \rightarrow Z$ be a function. Suppose $\Q^*(F)=0$. Then $\Q_2(F)=0$.
\end{restatable}

We finally show there exists a relation problem which can be solved without communication using non-signaling correlation, and requires to solve $\Omega(n)$ qubits communication for entanglement-assisted communication models. Let us denote by $\Q^{\mathcal{NS}}(R)$ quantum communication complexity with non-signaling correlation and by $\Q_2^*(R)$ quantum 2-way communication complexity with shared entanglement.

\begin{restatable}{theorem}{restateCHSHThm}\label{thm:chsh}
    There exists a relation $R$ such that $\Q^{\mathcal{NS}}(R)=0$ and $\Q_2^*(R)=\Theta(n)$.
\end{restatable}

Our results imply that the strength of correlations can completely change the complexity of quantum communication.

\subsection*{Proof strategies}

To prove the lower bound in \cref{thm:main}, we consider a reduction from a 2-way quantum communication protocol to a zero communication protocol and the parallel repetition theorem. The reduction consists of quantum teleportation \cite{BBC+93}, making the parties check that shared randomness meets their communication protocols, and replacing shared entanglement by a maximally mixed state. This reduction or removal of quantum communication makes the success probability exponentially worse in the communication amount. However, if the parallel repetition theorem holds with exponential decay, the success probability is at most an exponentially small value, and thus $\Omega(n)$ qubit communication is required in the original 2-way quantum communication protocol.

We show \cref{thm:function} by considering a simulation of quantum circuits to solve non-local relations without communication. To prove \cref{thm:chsh}, we consider parallel repetition of CHSH games \cite{clauser1969proposed}, which are known that their input-output relation does not violate causality and thus solved by non-signaling correlation \cite{popescu1994quantum}. To prove the lower bound, we apply a direct product theorem for the entanglement-assisted quantum communication complexity shown by Jain and Kundu \cite{jain2021direct}.

\subsection*{Related work}

It is known that $c$-bit classical messages can be guessed without communication using uniformly random bits with probability $2^{-c}$. Laplante, Lerays and Roland (see the proof of Theorem 5 in \cite{laplante2012classical}) showed, under entanglement assistance, quantum message can be also guessed with probability $2^{-2c}$ using quantum teleportation. Our technique further implies quantum message can be guessed with probability $2^{-4c}$ without the entanglement assistance using quantum teleportation and replacing maximally entangled states by maximally mixed states.

Newman's theorem \cite{New91} states that, by allowing constant error penalty, arbitrary shared randomness can be replaced by $O(\log n)$ bits private randomness. A line of work \cite{shi2005tensor,gavinsky2006bounded,Gav06,jain2008optimal,jain2008direct,gavinsky2009classical,AG23} investigated some trade-offs between the amount of shared entanglement and communication. Most recently, Arunachalam and Girish \cite{AG23} showed some fine-grained trade-offs for partial functions (i.e., there is a unique solution for some promised inputs). By a simple reduction, our result implies that there exist total relation problems (i.e., there are multiple solutions for all inputs) such that they can be solved with no communication if shared entanglement of $\Omega(n)$ EPR pairs is allowed, and $\Omega(n)$ quantum 2-way communication is required even if shared entanglement of $o(n)$ EPR pairs is allowed. Therefore, our result exhibits an asymptotically tight trade-off between the amount of entanglement and quantum 2-way communication (a quantum 2-way communication model is a stronger communication model than classical (and quantum) SMP (and 1-way) communication models). However, when we consider function problems which have upper bounds of zero communication in entanglement-assisted models, the problems can be solved without entanglement and communication (\cref{thm:function}). See \cref{tab:result2} for a comparison. Our result and the line of work \cite{shi2005tensor,gavinsky2006bounded,Gav06,jain2008optimal,jain2008direct,gavinsky2009classical,AG23} imply that a quantum analog of Newman's theorem does not hold.

\begin{table}[htbp]
    \centering
    \begin{tabular}{cc|c|c|c|c}
         & Problem & Type & Input size & Upper bound & Lower bound\\
         \hline \hline
         & $\oplus^k$-Forrelation & \begin{tabular}{c} partial \\ function \end{tabular} & $kn$ & \begin{tabular}{c} $\Tilde{O}(k^5 \log^3 n)$ EPR pairs \& \\ $\Tilde{O}(k^5 \log^3 n)$ com. in $\Q^{||}$ \end{tabular} & \begin{tabular}{c} $O(k)$ EPR pairs \& \\ $\Omega(n^\frac{1}{4})$ com. in $\C^2$ \end{tabular} \\ \hline
         & $\oplus^k$-Boolean HM & \begin{tabular}{c} partial \\ function \end{tabular} & $kn$ & \begin{tabular}{c} $\Tilde{O}(k \log n)$ EPR pairs \& \\ $\Tilde{O}(k \log n)$ com. in $\C^{||}$ \end{tabular} & \begin{tabular}{c} $O(k)$ EPR pairs \& \\ $\Omega(k \sqrt{n})$ com. in $\C^{1}$ \end{tabular} \\ \hline 
         & " & " & " & " & \begin{tabular}{c} $O(k)$ EPR pairs \& \\ $\Omega(k n^\frac{1}{3})$ com. in $\Q^{||}$ \end{tabular} \\ \hline 
         & \begin{tabular}{c} Parallel MSGs \\ (\cref{thm:main}) \end{tabular} & \begin{tabular}{c} total \\ relation \end{tabular} & $O(n)$ & \begin{tabular}{c} $2n$ EPR pairs \& \\ zero communication \end{tabular} & \begin{tabular}{c} $o(n)$ EPR pairs \& \\ $\Omega(n)$ com. in $\Q^2$ \end{tabular} \\ 
         \hline & \begin{tabular}{c} \cref{thm:function} \end{tabular} & \begin{tabular}{c} any \\ function \end{tabular} & $n$ & \begin{tabular}{c} $\infty$ EPR pairs \& \\ zero communication \end{tabular} & \begin{tabular}{c} zero EPR pairs \& \\ zero com. \end{tabular} \\
    \end{tabular}
    \caption{Trade-offs of entanglement and communication from our result and results in \cite{AG23}. For the $\oplus^k$-Forrelation problem, $k$ satisfies $o(n^\frac{1}{50})$. We abbreviate communication as ``com.''. The row of Type indicates corresponding problems are total/partial functions/relations. We denote by $\C^{||}$ classical SMP models, by $\C^{1}$ classical 1-way communication models, by $\C^{2}$ classical 2-way communication models, by $\Q^{||}$ quantum SMP models and by $\Q^{2}$ quantum 2-way communication models. In the classical and quantum SMP models, prior entanglement is allowed between two parties who have inputs. The second and third columns consider the same problem ($\oplus^k$-Boolean HM problem).}
    \label{tab:result2}
\end{table}

A non-local game which has a perfect quantum strategy plays an important role in the study of $\mathsf{MIP}^*$ \cite{cleve2004consequences,ji2021mip}. Such a non-local game that admits a perfect quantum strategy can be generalized and characterized in some ways \cite{cleve2014characterization,arkhipov2012extending}.

Jain and Kundu \cite{jain2021direct} proved a direct product theorem for the entanglement-assisted quantum communication complexity which we apply to prove the lower bound in \cref{thm:chsh}. The direct product theorem (Theorem 1 in \cite{jain2021direct}) is related to \cref{thm:main} but implies nothing if the quantum value of non-local games is 1 (i.e., there exists a perfect quantum strategy for non-local games). 

Recently, Braverman, Khot, and Minzer \cite{braverman2025parallel} proved the parallel repetition theorem for the GHZ game \cite{greenberger1989going} with exponential decay. Since there exists a perfect quantum strategy for the GHZ game, we can obtain a similar result for a 3-party communication protocol using the same proof strategy.

\section{Preliminaries}
We assume that the readers are familiar with the standard notation of quantum computing. We refer to \cite{NC10,Wat18,dW19} for standard references in quantum computing, to \cite{kushilevitz1996communication,rao2020communication} for references in communication complexity.
We also refer to \cite{BBT05,BCMdW10} for surveys of non-locality and quantum communication complexity, and to \cite{Raz10} for a survey of the parallel repetition theorem.

\subsection{Communication complexity}

Let us recall some basic definitions of quantum communication complexity.

In this paper, we consider a relation $R \subseteq X \times Y \times Z \times W$. In our setting, Alice receives an input $x \in X$ (unknown to Bob), and Bob receives an input $y \in Y$ (unknown to Alice). The (bounded error) communication complexity of $R$ is defined as the minimum cost of classical or quantum communication protocols to compute $z \in Z$ for Alice and $w \in W$ for Bob such that $(x,y,z,w) \in R$ with high probability, say $\frac{2}{3}$.

In the standard definition, Alice and Bob output the same value. In our definition, Alice and Bob may output different values. Our definition is a generalization of the standard setting, and one can recover the standard definition from our definition by requiring $z=w$ in the outputs to meet a relation $R$.

We denote by $\C_2^{pub}(R)$ the classical 2-way communication complexity with shared randomness for $R$, by $\Q_{2}^{pub}(R)$ the quantum 2-way communication complexity without shared entanglement and with public randomness for $R$ and by $\Q^{*}_2(R)$ the quantum 2-way communication complexity with shared entanglement for $R$.

\subsubsection*{Construction of non-local relations from communication protocols}

We can convert communication complexity problems to non-locality problems. See Section 4 in \cite{BCMdW10}. 

Consider a communication complexity problem $f$ whose quantum one-way communication complexity is $q$ and classical communication complexity is $c$. Suppose that it can be solved if Alice starts with some initial state $\ket{k}$ (the value of $k$ being known before the protocol to both Alice and Bob from the setting of communication complexity) as follows: she carries out a transformation $U_A(x)$ on this state, sends it to Bob who carries out a transformation $U_B(y)$ depends on his input $y$ and then measures in the computational basis. The probability of obtaining result $l$ is thus $|\bra{l} U_B(y) U_A(x)\ket{k}|^2$. From the knowledge of $l$, $k$, and $y$, Bob can obtain the value of the function $f(x,y)$.

Now let us consider the following process: Alice and Bob share a maximally entangled state $|\psi\rangle = 2^{-q/2}\sum_{i=0}^{2^q-1} \ket{i}\ket{i}$; Alice applies $U_A(x)^T$ (where $T$ is the transposition) and measures in the computational basis. Bob applies $U_B(y)$ and measures in the computational basis. Suppose that Alice obtains outcome $k$ and Bob obtains outcome $l$. The probability of finding these joint outcomes is $P(k,l|x,y)=|\bra{l}\bra{k}U_B(y)  U_A(x)^T \ket{\psi}|^2=2^{-q}|\bra{l} U_B(y) U_A(x)\ket{k}|^2$. This non-locality problem required at least $c$ bits of communication to solve by classical communication models by a reduction.

Let us recall the definitions of the non-local Deutsch-Jozsa problem and the non-local Hidden Matching problem, which can be constructed from the Distributed Deutsch-Jozsa problem \cite{BCW98} and the Hidden Matching problem \cite{BYJK08} respectively. Let us denote by $a\cdot b=\sum_i a_ib_i$ the inner product between bit strings $a$ and $b$, and by $a\oplus b$ the bitwise XOR of $a$ and $b$: the $i$th bit of $a \oplus b$ is $ a_i \oplus b_i$.

\begin{definition}[Non-local Deutsch-Jozsa problem \cite{BCT99}]
    For inputs $x,y \in \{0,1\}^n$ that satisfy the promise: either $x=y$, or $|x \oplus y|= n/2$. The task is to output $a \in \{0,1\}^{\log n}$ by Alice and $b\in \{0,1\}^{\log n}$ by Bob such that when $x=y$ then $a=b$, and when $|x \oplus y|= n/2$ then $a \neq b$.
\end{definition}

\begin{definition}[Non-local Hidden Matching problem, Section 6 in \cite{BBT05}]
    Assume that $n=2^m$, so we can index the numbers between 1 and $n$ with $m$-bit strings. Alice receives a string $x\in \{0,1\}^n$. Bob receives a perfect matching $M$ on $\{1,\ldots,n\}$ (i.e., a partition into $n/2$ disjoint pairs). Alice must give as output some $k\in \{0,1\}^m$. Bob must give as output a matching $(i,j)\in M$ and $l \in \{0,1\}^m$. Alice and Bob's output must satisfy $(i \oplus j) \cdot(k \oplus l) = x_i \oplus x_j$.
\end{definition}

\subsubsection*{Newman's theorem}

Newman \cite{New91} showed that any public-coin randomized communication protocol can be converted into a small private-coin randomized protocol by allowing for a small constant error. For a function $F:X \times Y \rightarrow \{0,1\}$ and $c>0$, let us denote by $\mathsf{C}_\gamma(F)$ the classical randomized communication complexity with private randomness and error $\gamma$, and let us denote by $\mathsf{C}_\gamma^{pub}(F)$ the classical randomized communication complexity with public unlimited randomness and error $\gamma$. 

\begin{fact}[\cite{New91}]
    Let $F: X \times Y \rightarrow \{0,1\}$ be a function. For every $\epsilon>0$ and $\delta>0$, $\mathsf{C}_{\epsilon+\delta} (F) \leq \mathsf{C}_\epsilon^{pub} (F) + O(\log (\log |X| + \log |Y|) + \log (\frac{1}{\delta}))$.
\end{fact}

To our knowledge, there has been no formal statement and proof of Newman's theorem in our setting, and thus let us show Newman's theorem holds for relation problems in our setting. It can be shown from Hoeffding's bound.

\begin{fact}[Hoeffding's bound \cite{Hoe63}]\label{fact:hoeffding}
    Let $x_1,\ldots,x_t$ be independent random variables such that $0 \leq x_i \leq 1$ for $i \in [t]$. Then,
    \[
        \mathrm{Pr}\left[ \left| \frac{1}{t} \left(\sum_{i=1}^t x_i - \mathbb{E} \left[ \sum_{i=1}^t x_i \right] \right) \right| \geq \delta \right] \leq 2e^{-2 \delta^2 t}.
    \]
\end{fact}

\begin{proposition}
    Let $R \subseteq X \times Y \times Z \times W$ be a relation. For every $\epsilon>0$ and $\delta>0$, $\mathsf{C}_{\epsilon+\delta} (R) \leq \mathsf{C}_\epsilon^{pub} (R) + O(\log (\log |X| + \log |Y|) + \log (\frac{1}{\delta}))$.
\end{proposition}

\begin{proof}
    We will show that any classical two-party communication protocol $\mathcal{P}$ with unlimited shared random bits between the two parties can be transformed into another classical two-party communication protocol $\mathcal{P}'$ in which Alice and Bob share only $O(\log n + \log (\frac{1}{\delta}))$ random bits while increasing the error by only $\delta$. Since the amount of randomness is small, by making Alice send all the random bits to Bob, we have the desired protocol. Let $\Pi$ be the probabilistic distribution of the shared randomness of the protocol $\mathcal{P}$.

    Let $V(x, y, r)$ be a random variable which is defined as the probability that $\mathcal{P}$’s output $(z,w) \in Z \times W$ on input $(x, y) \in X \times Y$ and random string $r$ shared by Alice and Bob satisfy $(x, y, z, w) \notin R$. Because $\mathcal{P}$ computes $R$ with $\epsilon$ error, we have $\mathbb{E}_{r \in \Pi} [V(x,y,r)] \leq \epsilon$ for all $(x,y)$. We will build a new protocol that uses fewer random bits, using the probabilistic method. Let $t$ be a parameter, and let $r_1,\ldots,r_t$ be $t$ strings. For such strings, let us define a protocol $\mathcal{P}_{r_1,\ldots,r_t}$ as follows: Alice and Bob choose $1 \leq i \leq t$ uniformly at random and then proceed as in $\mathcal{P}$ with $r_i$ as their common random string. We now show that there exist strings $r_1,\ldots,r_t$ such that $\frac{1}{t} \sum_{i=1}^t [V(x,y,r_i)] \leq \epsilon + \delta$ for all $(x,y)$. For this choice of strings, the protocol $\mathcal{P}_{r_1,\ldots,r_t}$ is the desired protocol.

    To do so, we choose the $t$ values $r_1,\ldots,r_t$ by sampling the distribution $\Pi$ $t$ times. Consider a particular input pair $(x,y)$ and compute the probability that $\frac{1}{t} \sum_{i=1}^t V(x,y,r_i) > \epsilon + \delta$. By Hoeffding's bound (\cref{fact:hoeffding}), since $\mathbb{E}_{r \in \Pi}[V(x,y,r)] \leq \epsilon$, we get
    \[
        \mathrm{Pr} \left[ \left( \frac{1}{t}  \sum_{i=1}^t V(x,y,r_i) - \epsilon \right) > \delta \right] \leq 2e^{-2\delta^2 t}.
    \]
    By choosing $t=O\left(\frac{ \log |X| + \log |Y| }{\delta^2}\right)$, this is smaller than $|X|^{-1} |Y|^{-1}$. Thus, for a random choice of $r_1,\ldots,r_t$, the probability that for some input $(x,y)$, $\frac{1}{t} \sum_{i=1}^t [V(x,y,r_i)] > \epsilon + \delta$ is smaller than $|X||Y| |X|^{-1} |Y|^{-1} = 1$. This implies that there exists a choice of $r_1,\ldots,r_t$ where for every $(x,y)$ the error of the protocol $\mathcal{P}_{r_1,\ldots,r_t}$ is at most $\epsilon + \delta$. Finally, note that the number of random bits used by the protocol $\mathcal{P}_{r_1,\ldots,r_t}$ is $\log t = O(\log (\log |X| + \log |Y|) + \log (\frac{1}{\delta}))$.
\end{proof}

\subsection{Non-local game and parallel repetition theorem}

Let us recall some basic definitions of non-local games and parallel repetition theorems. A non-local game is played between two provers or players called Alice and Bob. The game consists of four finite sets $X, Y, A, B$, a probability distribution $P$ over $X \times Y$ and a predicate $V : X \times Y \times A \times B \rightarrow \{0, 1\}$. All parties know $X, Y, A, B, P, V$. Intuitively, $X$ is the set of possible questions for the first prover, $Y$ the set of possible questions for the second prover, $A$ the set of possible answers of the first prover, and $B$ the set of possible answers of the second prover. The distribution $P$ is used to generate questions for the two provers, and the predicate $V$ is used to accept or reject after the answers from both provers are obtained.

The game proceeds as follows. A pair of questions $(x, y) \in X \times Y$ is sampled from $P$, and sends $x$ to Alice and $y$ to Bob. Each party knows only the question addressed to him or her, and is not allowed to communicate with each other. Alice calculates some function and outputs $a = a(x) \in A$ and Bob calculates another function and outputs $b = b(y) \in B$. They win if $V (x, y, a, b) = 1$. The classical value of the game is the maximum probability of success that Alice and Bob can achieve, where the maximum is taken over all protocols $(a, b)$; that is, the classical value of the game is
\[
    \max_{a,b} \Exp_{(x,y)} [V (x, y, a(x), b(y))]
\]
where the expectation is taken with respect to the distribution $P$. Note that a setting where Alice and Bob share some randomness according to some distribution can be considered, and the classical value of the game is invariant in the setting. This is because, for a shared value $r$ chosen from a probabilistic distribution over $R$, there exist $r' \in R$ and functions $a', b'$ such that
\[
    \Exp_{x,y,r} [V (x, y, a(x,r), b(y,r))] \leq  \Exp_{x,y} [V (x, y, a(x,r'), b(y,r'))] = \Exp_{x,y} [V (x, y, a'(x), b'(y))].
\]
In other words, a randomized strategy can simply be viewed as some convex mixture of deterministic strategies, and the parties can select the optimal deterministic strategy, and the average winning probability obviously cannot be larger than the maximum winning probability over all deterministic strategies.

Roughly speaking, the parallel repetition of a non-local game $G$ is a game where the parties try to win simultaneously $n$ copies of $G$. The parallel repetition game is denoted by $G^{\otimes n}$. More precisely, in the game $G^{\otimes n}$, we consider questions $x = (x_1,\dots,x_n) \in X_n$, $y=(y_1,\dots,y_n) \in Y_n$, where each pair $(x_i, y_i) \in X \times Y$ is chosen independently according to the original distribution $P$. For outputs $a = (a_1,\dots,a_n) = a(x) \in A_n$ and $b = (b_1,\dots,b_n) = b(y) \in B_n$, they win if they win simultaneously on all $n$ coordinates, that is, if for every $i$, we have $V(x_i,y_i,a_i,b_i) = 1$.

It is known that the following parallel repetition theorem for classical values holds.

\begin{theorem}[Parallel repetition theorem \cite{Raz98,Hol09}]\label{thm:parallel_repetition}
    For any non-local game $G$ whose classical value is $1-\epsilon$ for $0<\epsilon\leq \frac{1}{2}$, the classical value of the game $G^{\otimes n}$ is at most 
    \[
        (1-\epsilon^c)^{\Omega(\frac{n}{s})},
    \]
    where $s=\log |A \times B| +1$ and $c$ is a universal constant.
\end{theorem}

A quantum strategy for a game $G = (P, X \times Y, A \times B, V)$ consists of an initial bipartite state $\ket{\psi} \in \mathcal{H}_A \otimes \mathcal{H}_B$ for finite-dimensional Hilbert spaces $\mathcal{H}_A$ and $\mathcal{H}_B$, a quantum measurement $\{M^{x}_{a}\}_{a \in A}$ for each $x \in X$ over $\mathcal{H}_A$, and a quantum measurement $\{M^{y}_{b}\}_{b \in B}$ each $y$ over $\mathcal{H}_B$. On a pair of inputs $(x,y)$ drawn from $P$, Alice performs her measurement corresponding to $x$ on her portion of $\ket{\psi}$, yielding an outcome $a$. Similarly, Bob performs his measurement corresponding to $y$ on his portion of $\ket{\psi}$, yielding outcome $b$. The results $a$ and $b$ are sent back to the referee, who evaluates if two answers and questions meet the predicate $V$. In summary, the quantum winning probability for the quantum strategy above is given by
\[
    \sum_{x,y} \pi(x,y) \sum_{a,b} \bra{\psi} M_a^x \otimes M_b^y \ket{\psi} V(x,y,a,b).
\]
The quantum value of a game $G$, denoted by $\omega_q(G)$, is the supremum of the winning probabilities over all the quantum strategies of Alice and Bob.

Jain and Kundu \cite{jain2021direct} proved the following direct product theorem for the entanglement-assisted interactive quantum communication complexity. 

\begin{theorem}[Part of Theorem 1 in \cite{jain2021direct}]\label{thm:dpt}
For any predicate $V$ on $(A\times B)\times(X\times Y)$ and any product probability distribution $P$ on the question set $X\times Y$, let $\mathcal{P}$ be an interactive entanglement-assisted 2-party quantum communication protocol for $G^{\otimes n}$ which has total communication $cn$. If $c < 1$, the success probability of $\mathcal{P}$ is at most
\[ \left(1-\frac{\nu}{2} + 2\sqrt{c}\right)^{\Omega\left(\frac{\nu^2n}{\log(|A|\cdot|B|)}\right)}\]
where $\nu=1-\omega_q (G)$.
\end{theorem}
Note that Theorem 1 in \cite{jain2021direct} is a more general argument including the case of more than 2 parties and $c \geq 1$.

\subsubsection*{Magic square game}

Let us recall the definition of the magic square game \cite{Ara02}.

Alice is asked to give the entries of a row $x \in \{1, 2, 3\}$ and Bob is asked to give the entries of a column $y \in \{1, 2, 3\}$. The winning condition is that the parity of the row must be even, the parity of the column must be odd, and the intersection of the given row and column must agree.

Inputs are drawn uniformly from $\{1, 2, 3\} \times \{1, 2, 3\}$, and the classical winning probability is at most $8/9$. On the other hand, there is a perfect quantum strategy for the game due to Mermin \cite{Mer90,Mer93}.

\subsubsection*{CHSH game}

CHSH game \cite{clauser1969proposed} is a non-local game over $X=Y=A=B=\{0,1\}$. The provers win if and only if $a \oplus b = x \land y$. The quantum value is $\frac{2+\sqrt{2}}{4}$ \cite{cirel1980quantum}. If non-signaling correlation is allowed (the provers can do anything as long as causality of information is not violated, i.e., information between them cannot travel faster than the speed of light), there exists a perfect strategy \cite{popescu1994quantum}.

\section{Proofs}

From \cref{thm:parallel_repetition} and since the classical value of the magic square game is 8/9, we have the following statement.
\begin{lemma}\label{lem:parallel_repetition_magic}
    Let $G$ be the magic square game. There exists a constant $c>0$ such that the classical value of $G^{\otimes n}$ is at most $2^{-cn}$.
\end{lemma}

As a warming-up, we first show that $\Omega(n)$ classical communication is required for the parallel repetition of the magic square games. Similar statements are explicitly stated in Corollary 14 in \cite{jain2021direct} and Fact 4 in \cite{bharti2023power}.

\begin{theorem}\label{thm:classical_lower_bound}
    Let $G$ be the magic square game. Let $\mathcal{P}$ be a $2$-way classical communication protocol with public randomness for $G^{\otimes n}$. Suppose that the success probability of $\mathcal{P}$ is $2^{-o(n)}$. Then, the communication amount of $\mathcal{P}$ is $\Omega(n)$, and thus $\C_2^{pub}(G^{\otimes n}) = \Omega(n)$.
\end{theorem}

\begin{proof}
    Let $c'>0$ be a constant and $\mathcal{P}$ be a classical 2-way communication protocol with public randomness for $G^{\otimes n}$ whose communication amount is $c' n$. Let $p$ be the success probability of $\mathcal{P}$. 

    Let us consider a communication protocol $\mathcal{P}'$ in which the parties share the original public coin in the protocol $\mathcal{P}$ and uniformly random $c'n$ bits, and have no communication. In the protocol $\mathcal{P}'$, Alice and Bob check that the shared uniform randomness indeed meets the $c' n$ bits communication protocol, and if not, they abort (and output anything). Let $\not \perp_A$ be the event that Alice does not abort during the protocol $\mathcal{P}'$. Let $\not \perp_B$ be the event that Bob does not abort during the protocol $\mathcal{P}'$.

    Let $W_{\mathcal{P}'}$ be the event that $\mathcal{P}'$ outputs valid values (wins all the games) without abortions. Then, from the definition of $\mathcal{P}$, we have
    \[
        \mathrm{Pr} [ W_{\mathcal{P}'} | \not \perp_A, \not \perp_B] = p.
    \]
    Since the shared uniform randomness and the original communication protocol $\mathcal{P}$ are independent, we also have 
    \[
        \mathrm{Pr} [ \not \perp_A, \not \perp_B ] = 2^{-c' n},
    \]
    which implies 
    \[
        \mathrm{Pr} [ W_{\mathcal{P}'} ] = 2^{-c' n} p.
    \]
    Considering abortions, the success probability of $\mathcal{P}'$ is at least $2^{-c' n} p$.
    From \cref{lem:parallel_repetition_magic} and since shared randomness does not change the classical value of games, the success probability of $\mathcal{P}'$ is also at most $2^{-cn}$. Taking $c > c' >0$ and combining the upper and lower bound of the success probability, $p \leq 2^{-(c-c')n}$. Taking the contraposition, we have the claim.
\end{proof}

Next, we prove a lower bound with quantum communication.

\begin{theorem}\label{thm:lower_bound}
    Let $G$ be the magic square game. Let $\mathcal{P}$ be a $2$-way quantum communication protocol without prior entanglement and with public randomness for $G^{\otimes n}$. Suppose that the success probability of $\mathcal{P}$ is $2^{-o(n)}$. Then, the communication amount of $\mathcal{P}$ is $\Omega(n)$.
\end{theorem}

\begin{proof}
Let the amount of communication of $\mathcal{P}$ be $c' n$ qubits for a constant $c' = \frac{c}{5}>0$ where $c$ is the constant in \cref{lem:parallel_repetition_magic}, and let the success probability of $\mathcal{P}$ be $p$. By quantum teleportation \cite{BBC+93}, we can replace a $1$-qubit message with a $2$-bit message assuming that an EPR pair is shared by the parties. By allowing shared $c' n$ EPR pairs and considering quantum teleportation, we can construct a classical communication protocol $\mathcal{P}'$ where parties share the public coin in $\mathcal{P}$ and additional $c' n$ EPR pairs, whose communication amount is $2 c' n$ bits and success probability is $p$.

Next, we construct a communication protocol $\mathcal{P}''$ in which the parties share the public coin, $c' n$ EPR pairs and $2c' n$ uniform randomness without communication, and the success probability is at least $p \cdot 2^{-2c'n}$. This is done by using the shared uniform randomness, Alice and Bob check that the randomness indeed meets the $2 c' n$ bits communication protocol, and if not, they abort and output anything (as the same way in the proof of \cref{thm:classical_lower_bound}).

Let us replace $c' n$ EPR pairs by a maximally mixed state. Let us denote
\[
    \ket{\mathsf{EPR}} = \frac{1}{\sqrt{2}}(\ket{00} + \ket{11}).
\]
A maximally mixed state over Alice and Bob 
\[
    \frac{1}{2^{2c'n}} I = \frac{1}{2^{c'n}} I \otimes \frac{1}{2^{c'n}} I
\]
is a separable state between Alice and Bob, and thus it can be produced by Alice and Bob without communication. Let us denote
\[
    \rho = \frac{1}{2^{2c'n}-1} \sum_i \ket{\psi_i} \bra{\psi_i}
\]
such that $\ket{\mathsf{EPR}}^{\otimes c' n}$ and $\{ \ket{\psi_i} \}_i$ form an orthogonal basis of the $2^{2c'n}$-dimensional Hilbert space.
Since 
\[
    \frac{1}{2^{2c'n}} I = \frac{1}{2^{2c'n}} \left( \ket{\mathsf{EPR}} \bra{\mathsf{EPR}} \right)^{\otimes c' n} + \left(1-\frac{1}{2^{2c'n}} \right) \rho,
\]
we have a communication protocol $\mathcal{P}'''$ with shared randomness, without shared entanglement and communication, and the success probability is at least $p \cdot 2^{- (2+2) c' n}$.

From \cref{lem:parallel_repetition_magic} and shared randomness does not change the classical values, we have $2^{-cn} \geq p \cdot 2^{-4 c' n}$, which implies $p \leq 2^{- (c - 4 c')\cdot n}$. Since $c - 4 c' = \frac{1}{5} c >0$, the success probability of $\mathcal{P}$ is $2^{-\Omega(n)}$. Taking the contraposition, we have the claim.

\end{proof}

Let us prove our main theorem:

\restateMainThm*

\begin{proof}[Proof of \cref{thm:main}]
    Let $G$ be the magic square game and we will regard $G^{\otimes n}$ as a relation. Since there exists a quantum perfect strategy for the magic square game, $\Q^*(G^{\otimes n})=0$. The input size of $G^{\otimes n}$ is $O(n)$, and from \cref{thm:lower_bound}, we have $\Q_2^{pub}(G^{\otimes n})=\Theta(n)$.
\end{proof}

As a complement result, we show that such the largest separation does not hold for a function rather than a relation.

\restateCompThm*

\begin{proof}[Proof of \cref{thm:function}]
    Since $Q^*(F)=0$, there exists an entanglement-assisted communication protocol $\mathcal{P}$ whose communication amount is 0 to compute $F$ with high probability. Without loss of generality, in the protocol $\mathcal{P}$, we assume that Alice applies a unitary operation $U_A$ on her input $x$ encoded in the computational basis and her ancillary system $\ket{\mathrm{ancilla}}_A$ and a part of shared entanglement possessed by her, and she measures a part of her quantum registers and outputs the measurement result. We also assume that Bob applies a unitary operation $U_B$ on his input $y$ in the computational basis and his ancillary system $\ket{\mathrm{ancilla}}_B$ and a part of shared entanglement possessed by him, and he measures a part of his quantum registers outputs the measurement result. As a quantum circuit, we can describe the protocol as in \cref{fig:circuit}. Given the protocol $\mathcal{P}$, the marginal distribution of Alice's measurement results depend only on $x$ and the marginal distribution of Bob's measurement results depend only on $y$.
    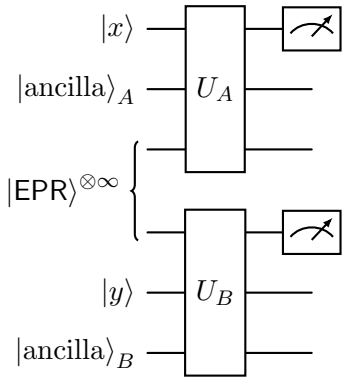
\begin{figure}[H]
        \centering
        \begin{quantikz}
            \lstick{$\ket{x}$} & \gate[3]{U_A} & \meter{}  \\
            \lstick{$\ket{\mathrm{ancilla}}_A$} & &  \\
            \lstick[2]{$\ket{\mathsf{EPR}}^{\otimes \infty}$} & & \\
            & \gate[3]{U_B} & \meter{}\\
            \lstick{$\ket{y}$}& & \\
            \lstick{$\ket{\mathrm{ancilla}}_B$} & &
        \end{quantikz}
        \caption{Circuit description of the entanglement-assisted protocol $\mathcal{P}$ to compute $F$.}	
        \label{fig:circuit}
    \end{figure}
    Let us construct a new communication protocol $\mathcal{P}'$ without shared entanglement and communication from the protocol $\mathcal{P}$. Not to confuse the readers, let us call the two parties in $\mathcal{P}'$ Carol and Dave who receive inputs $x$ and $y$ respectively. Carol simulates the marginal distribution of Alice's measurement results. Note that since the results only depend on $x$ and the maximally mixed state which is a reduced state on a part of entanglement Alice possesses, this simulation can be done exactly by Carol without communication. Then, Carol outputs a measurement result whose probability is the highest among all measurement results. Note that this result is uniquely determined because the original protocol $\mathcal{P}$ computes the function $F$ with high probability. Similarly, Dave simulates the marginal distribution of the Bob's measurement results and outputs a measurement result whose probability is the highest among all measurement results. This protocol computes $F$ with probability $1$. Therefore, we have $\Q_2(F)=0$, which concludes the proof.  
\end{proof}

Using \cref{thm:dpt}, we finally show the maximum separation of quantum communication complexity with non-signaling correlation and shared entanglement.

\restateCHSHThm*

\begin{proof}[Proof of \cref{thm:chsh}]
    For the CHSH game $G$, let us consider the communication complexity of $G^{\otimes n}$. If non-signaling correlation is allowed, there exists a perfect strategy, and thus $\Q^{\mathcal{NS}}(G^{\otimes n}) = 0$. Moreover, since the question distribution of CHSH game is a product distribution and its quantum value is $\frac{2+\sqrt{2}}{4}$, we can apply \cref{thm:dpt} and have $\Q^*_2(G^{\otimes n}) = \Theta(n)$.
\end{proof}

\section*{Acknowledgments}
The authors are grateful to Rahul Jain for useful discussions and clarifying implications from \cite{jain2021direct}. AH also thanks Ryuhei Mori, Harumichi Nishimura and Daiki Suruga for helpful conversations, Takashi Yamakawa for spotting the proof of Theorem 5 in \cite{laplante2012classical} as a related technique.

AH was supported by JSPS KAKENHI grant No.~24H00071, 25K24674. FLG was supported by JSPS KAKENHI grants Nos.~JP20H05966, 20H00579, 24H00071, 25K24674, MEXT Q-LEAP grant No.~JPMXS0120319794 and JST CREST grant No.~JPMJCR24I4. 
Most of this work was done while AM was affiliated with Aalto University. AM was
also partially supported by the Helsinki Institute for Information Technology
(HIIT) and the Research Council of Finland, Grant 359104.

\bibliographystyle{alpha}
\bibliography{ref}

\end{document}